\documentclass[10pt,conference]{IEEEtran}

\usepackage{latexsym}
\usepackage{bbold}
\usepackage{amsmath}
\usepackage{amssymb}
\usepackage{amsthm}
\usepackage{color}
\usepackage{amsmath}
\usepackage{bbm}
\usepackage{epsfig}
\usepackage{hyperref}
\usepackage{enumerate}

\newtheorem{theorem}{Theorem}

\newtheorem{definition}[theorem]{Definition}

\newcommand{\kk}[1] {{\bf #1}}
\begin{document}

\title{Allocations for Heterogenous Distributed Storage}

\author{\IEEEauthorblockN{Vasileios Ntranos}
\IEEEauthorblockA{University of Southern California \\ 
Los Angeles, CA 90089, USA \\
ntranos@usc.edu}
\and
\IEEEauthorblockN{Giuseppe Caire}
\IEEEauthorblockA{University of Southern California \\ 
Los Angeles, CA 90089, USA \\
caire@usc.edu}
\and
\IEEEauthorblockN{Alexandros G. Dimakis}
\IEEEauthorblockA{University of Southern California \\ 
Los Angeles, CA 90089, USA \\
dimakis@usc.edu}
}

\maketitle

\begin{abstract}
We study the problem of storing a data object in a set of data nodes that fail independently with given probabilities. Our problem is a natural generalization of 
a homogenous storage allocation problem where all the nodes had the same reliability and is naturally motivated for peer-to-peer and cloud storage systems with different types of nodes. 
Assuming optimal erasure coding (MDS), the goal is to find a storage allocation (i.e, how much to store in each node) to maximize the probability of successful recovery. 
This problem turns out to be a challenging combinatorial optimization problem. 
In this work we introduce an approximation framework based on large deviation inequalities and convex optimization. We propose two approximation algorithms and study the asymptotic performance of the resulting allocations.
SUBMITTED TO ISIT 2012.
\end{abstract}

\section{Introduction}
We are interested in heterogenous storage systems where storage nodes have different reliability parameters. This problem is relevant for heterogenous peer-to-peer storage networks and cloud storage systems that use multiple types of storage devices, \textit{e.g.} solid state drives along with standard hard disks. 
We model this problem by considering $n$ storage nodes and a data collector that accesses a random subset $\kk r$ of them. 
The probability distribution of $\kk r \subseteq \{1, \dots, n\}$ models random node failures and we assume that node $i$ fails independently with probability $1-p_{i}$.
The probability of a set $\kk r$ of nodes being accessed is therefore: 
\begin{equation}
{\mathbb P}(\kk r) = \prod_{i\in \kk r} p_{i}\prod_{j\notin \kk r}{(1-p_{j})}.
\end{equation}

Assume now that we have a single data file of unit size that we wish to code and store over these nodes to {\it maximize the  probability of recovery} after a random set of nodes fail.  
The problem becomes trivial if we do not put a constraint on the maximum size $T$ of coded data 
and hence, we will work with a maximum storage budget of size $T<n$: If $x_{i}$ is the amount of coded data stored in node $i$, then $\sum_{i=1}^{n}x_{i} \leq T$.
We further assume that our file is {\it optimally coded}, in the sense that
successful recovery occurs whenever the total amount of data accessed by the data collector is at least the size of the original file. This is possible in practice when we use Maximum Distance Separable (MDS) codes \cite{alloc}. 
The probability of successful recovery for an allocation $(x_{1}, \dots, x_{n})$ can be written as 
$$P_{s} = {\mathbb P} \left[ \sum_{i \in \kk r} x_{i} \geq 1 \right] = 
\sum_{\kk r \subseteq \{1, \dots, n\}}{\mathbb P}(\kk r)\;{\mbox{\Large $\mathbbm{1}$}}\Big\{\sum_{i\in \kk r} x_{i}\geq 1\Big\}$$
where ${\mbox{$\mathbbm{1}$}}\{\cdot\}$ is the indicator function. 
${\mbox{$\mathbbm{1}$}}\{S\} = 1$ if the statement $S$ is true and zero otherwise.

A more concrete way to see this problem is by introducing a $Y_{i}$ $\sim$ Bernoulli($p_{i}$) random variable for 
each storage node: $Y_{i}=1$ when node $i$ is accessed by the data collector and $Y_{i}=0$ when node $i$ has failed.  
Define the random variable 
\begin{equation}
Z = \sum_{i=1}^{n}x_{i}Y_{i}
\end{equation} 
where $x_{i}$ is the amount of data stored in node ${i}$. 
Then, obviously, we have $P_{s}={\mathbb P}[Z \geq 1]$.

Our goal is to find a storage allocation $(x_{1}, \dots, x_{n})$, that maximizes the probability of successful recovery, or equivalently,
minimizes the probability of failure, ${\mathbb P}[Z < 1]$. 
\section{Optimization Problem} 
 
Put in optimization form, we would like to find a solution to the following problem.
\begin{equation*}
\begin{aligned}
Q1: \;\;\;\;\;\;\;\;\;& \underset{\displaystyle x_{i}}{\text{minimize}} 
&  & \sum_{\kk r \subseteq \{1, \dots, n\}}{\mathbb P}(\kk r)\;{\mbox{\Large $\mathbbm{1}$}}\Big\{\sum_{i\in \kk r} x_{i} < 1\Big\} \\
& \text{subject to:}
& & \;\;\;\;\sum_{i=1}^{n}x_{i} \leq T  \\
& & &\;\;\;\; x_{i} \geq 0 , \; i = 1, \ldots, n.
\end{aligned}
\end{equation*}

Authors in \cite{alloc} consider a special case of problem $Q1$ in which $p_{i}=p$, $\forall i$. 
Even in this symmetric case the problem appears to be very difficult to solve due to its 
non-convex and combinatorial nature. 
In fact, even for a given allocation $\{x_i\}$ and parameter $p$, computing the objective function is computationally intractable ($\#P$-$hard$ , See \cite{alloc}).

%
A very interesting observation about this problem follows directly from Markov's Inequality: 
${\mathbb P}[Z \geq 1] \leq E[Z] = pT$. If $pT<1$, then the probability of successful recovery is 
bounded away from 1. This has motivated the definition of a region of parameters for which high probability of recovery
is possible: $R_{HP}=\{(p,T): pT\geq 1\}$. The budget $T$ should be more than $1/p$ if we want to aim for high 
reliability and the authors in \cite{alloc} showed that in the above region of parameters, {\it maximally spreading} the
budget to all nodes (i.e, $x_{i}=T/n$, $\forall i$) is an asymptotically optimal allocation as $n\rightarrow \infty$.

In the general case, when the node access probabilities, $p_{i}$, are not equal,
one could follow similar steps to characterize a region of high probability of recovery. Markov's Inequality yields:
$${\mathbb P}[Z \geq 1] \leq E[Z] = \sum_{i=1}^{n} x_{i} p_{i} = \kk p^{T} \kk x$$
where $\kk p = [p_{1}, p_{2}, \dots, p_{n}]^{T}$  and $\kk x = [x_{1}, x_{2}, \dots, x_{n}]^{T}$. 
If we don't want ${\mathbb P}[Z \geq 1]$ to be bounded away from $1$ we have to require now that $\kk p^{T} \kk x \geq 1$.
We see that in this case, high reliability is not a matter of sufficient budget, as it depends on the allocation $\bf x$ itself. 

Let $S(\kk p, T)=\left\{ \kk x \in {\mathbb R}_{+}^{n}\; :\; \kk p^{T} \kk x \geq 1 , \kk 1^{T} \kk x \leq T \right\}$ be the set of all allocations $\kk x$ with a given budget constraint $T$ that
satisfy $\kk p^{T} \kk x \geq 1$ for a given  $\bf p$. We call these allocations {\it reliable} for a system with parameters $\kk p$, $T$, in the sense that the resulting probability of successful recovery is not bounded away from 1.
Then the {region of high probability of recovery} can  be defined 
as the region of parameters $\kk p$, $T$, such that the set $S(\kk p, T)$ is non-empty.
$${\cal R}_{HP} = \left\{(\kk p,T) \in {\mathbb R}^{n+1}_{+}\;:\;  S(\kk p, T) \neq \emptyset  \right\}$$

This generalizes the region described in \cite{alloc}.
If all $p_{i}$'s are equal then the set $S(\kk p, T)$ is non-empty when $\kk p^{T} \kk x = pT \geq 1$. 
In the general case, the minimum budget such that $S(\kk p, T)$ is non-empty is $T = 1/p_{max}$,
with $p_{max}=\max\{p_{i}\}$, and
$S(\kk p, 1/p_{max})$ contains only one allocation $\kk x_{p^{-1}_{max}}: x_{j} = \frac{1}{p_{max}}\,,\; j = arg \, \max_{i}\{p_{i}\},\,
x_{i} = 0 \,,\; \forall i\neq j$. 

Even though ${\cal R}_{HP}$ provides a lower bound on the minimum budget $T$ required 
to allocate for high reliability, it doesn't provide any insights on how to design allocations
that achieve high probability of recovery in a distributed storage system. This motivates us to move one step further and 
define a region of $\epsilon$-optimal allocations in the next section.

\section{The region of $\epsilon$-optimal allocations}

We say that an allocation $(x_{1}, x_{2}, \dots, x_{n})$ is $\epsilon$-optimal if the corresponding probability of 
successful recovery, $\mathbb{P}[Z \geq 1]$, is greater than $1-\epsilon$. 

Let ${\cal E}_{n}(\kk p, T, \epsilon) = \{\;\kk x \in {\mathbb R}_{+}^{n}\; 
:\; \mathbb{P}[Z < 1] \leq \epsilon \;
 ,\;\; \kk 1^{T} \kk x \leq T \;\}$ be the set of all $\epsilon$-optimal allocations. Note that 
 if we could {\it efficiently} characterize this set for all problem parameters, we would be able to 
 solve problem $Q1$ exactly: Find the smallest $\epsilon$ such that ${\cal E}_{n}(\kk p, T, \epsilon)$ is non-empty.
 
In this section we will derive a sufficient condition for an allocation to be $\epsilon$-optimal
and provide an efficient characterization for a region ${\cal H}_{n}\subseteq{\cal E}_{n}(\kk p, T, \epsilon)$.
We begin with a very useful lemma.

\newtheorem{thm2}{Theorem}
\newtheorem{lemma1}[thm2]{Lemma} 

\begin{lemma1}{(Hoeffding's Inequality \cite{Hoeffding,mitzbook}})\\
Consider the random variable $W = \sum_{i=1}^{n} V_{i}$, where $V_{i}$ are independent almost surely bounded random variables with 
${\mathbb P}\left(V_{i}\in [a_{i},b_{i}]\right)=1$.  
Then,
$${\mathbb P}\Big[W \leq E[W] - n\delta\Big] \leq exp\left\{-\frac{2n^{2}\delta^{2}}{\sum_{i=1}^{n}(b_{i}-a_{i})^{2}}\right\}$$
for any $\delta>0$.
\label{lem1}
\end{lemma1}

We can use Lemma \ref{lem1} to upper bound the probability of failure, $\mathbb{P}[Z<1]\leq\mathbb{P}[Z\leq1]$, for an arbitrary allocation, 
since $Z=\sum_{i=1}^{n}x_{i}Y_{i}$ can be seen as the sum of $n$ independent almost surely bounded random variables $V_{i} = x_{i}Y_{i}$, with
${\mathbb P}\big(V_{i}\in [0,x_{i}]\big)=1$. 
Let $\delta = \big(\sum_{i=1}^{n}x_{i}p_{i}-1\big)/n$ and require $\delta > 0 \Leftrightarrow  \sum_{i=1}^{n}x_{i}p_{i}>1$.  
Lemma \ref{lem1} yields: 

\begin{equation}
\begin{aligned}
\mathbb{P}[Z<1] &\leq   exp\left\{-\frac{2\big(\sum_{i=1}^{n}x_{i}p_{i}-1\big)^{2}}{\sum_{i=1}^{n}x_{i}^{2}}\right\} ,\;\;\; \sum_{i=1}^{n}x_{i}p_{i}>1.\\
\end{aligned}
\label{bound}
\end{equation}
Notice that the constraint  $\sum_{i=1}^{n}x_{i}p_{i} > 1$ requires the allocation $(x_{1}, x_{2}, \dots, x_{n})$ 
to be reliable and $S(\kk p, T)\neq \emptyset$.  

In view of the above, a sufficient condition for a strictly reliable allocation to be $\epsilon$-optimal is the following.

\begin{equation}
\begin{aligned}
exp\left\{-\frac{2\big(\sum_{i=1}^{n}x_{i}p_{i}-1\big)^{2}}{\sum_{i=1}^{n}x_{i}^{2}}\right\} \leq \epsilon \;\;\;\Longleftrightarrow\;\;\;\\
{||\kk x||}_{2}\sqrt{\frac{\ln1/\epsilon}{2}\;} \leq \kk p^{T}\kk x - 1 \,, \;\;\;\;\;\; \kk p^{T}\kk x > 1
\end{aligned}
\label{suffcond}
\end{equation}

We say that all allocations satisfying the above equation are {\it Hoeffding $\epsilon$-optimal}, due to
the use of Hoeffding's Inequality in Lemma \ref{lem1}.

\begin{definition}{``The Region of Hoeffding $\epsilon$-optimal allocations''}

\begin{equation}
\begin{aligned}
{\cal H}_{n}(\kk p, T, \epsilon) =   
\Bigg\{& \kk x \in {\mathbb R}^{n}_{+} \;:\;  \kk p^{T}\kk x > 1 ,\; \kk 1^{T}\kk x \leq T ,\\ & \left. \;
{||\kk x||}_{2}\sqrt{\frac{\ln1/\epsilon}{2}\;} \leq \kk p^{T}\kk x - 1\right\}
\end{aligned}
\end{equation}

\end{definition}

The above region is strictly smaller ${\cal E}_{n}( \kk p, T, \epsilon)$ for any finite $n$, 
because the bound in (\ref{bound}) is not generally tight. 
However,  ${\cal H}_{n}( \kk p, T, \epsilon)$ is a convex set:
Equation (\ref{suffcond}) can be seen as a second order cone constraint on the allocation $\kk x \in {\mathbb R}^{n}_{+}$. 

\newtheorem{tm}{Theorem}

\begin{tm}
The region of Hoeffding $\epsilon$-optimal allocations ${\cal H}_{n}(\kk p, T, \epsilon)$ is convex in \kk x.
\end{tm}


This interesting result allows us to formulate and efficiently solve optimization problems over ${\cal H}_{n}(\kk p, T, \epsilon)$.
Finding the smallest $\epsilon^{*}$ such that ${\cal H}_{n}(\kk p, T, \epsilon)$ is non-empty will produce an $\epsilon^{*}$-optimal solution to 
problem $Q1$.

\subsection{Hoeffding Approximation of $Q1$}
If we fix $\kk p, T, n$ as the problem parameters, then the following optimization problem 
can be solved efficiently, to any desired accuracy $1/\alpha$, by solving a sequence of ${\cal O}(\log\alpha)$ convex feasibility problems 
(bisection on $\epsilon$).

\begin{equation*}
\begin{aligned}
&H1: \;\;\;\;\;\;\;\;&  &\underset{  \kk x, \epsilon }{\min} & & \epsilon & \\
& &&\text{s.t.:}& &\kk x \in {\cal H}_{n}(\kk p, T, \epsilon)&\\
\end{aligned}
\end{equation*}

We will see next that if $T$ is sufficiently large, $\epsilon^{*}$ goes to zero exponentially fast as $n$ grows, and hence the solution to the aforementioned problem is asymptotically optimal.


\subsection{Maximal Spreading Allocations and the Asymptotic Optimality of H1}

First, we will focus on maximal spreading allocations, ${\bf x}_{T}^{n} \triangleq \{{\bf x} \in {\mathbb R}^{n} : x_{i} = T/n\,\}$, and 
derive their asymptotic optimality for $Q1$, in the sense that ${\mathbb P}[Z<1] \rightarrow 0$, as $n\rightarrow \infty$. 
Let ${\bar p} = \frac{1}{n}\sum_{i=1}^{n}p_{i}$ be the average access probability across all nodes. We have the following lemma.

\begin{lemma1}
If $T > 1/ \bar  p$,\, for any $\epsilon > 0$, $\exists n_{\epsilon}$: ${\bf x}_{T}^{n} \in {\cal H}_{n}(\kk p, T, \epsilon)$, for all $n\geq n_{\epsilon}$.
\label{thm2}
\end{lemma1}

\begin{proof}
This follows directly from the definition of ${\cal H}_{n}(\kk p, T, \epsilon)$:
 $n_{\epsilon} = {{\frac{\ln1/\epsilon}{2(\bar p - 1/T)^{2}}\;}}$. 
%
%
\end{proof}

The above lemma establishes the asymptotic optimality of maximal spreading allocations through the following corollary.

\newtheorem{thm21}{Theorem}
\newtheorem{cor}[thm21]{Corollary}

\begin{cor}
The probability of failed recovery, $P_{e} \triangleq {\mathbb P}[Z<1]$, for a maximal spreading allocation   is
$P_{e} \leq e^{-2n(\bar p - 1/T)^{2}}$.
When $T>1/\bar p$, $P_{e}\rightarrow 0$, as $n \rightarrow \infty$.
\label{cor:maxspread}
\end{cor} 

The fact that ${\cal H}_{n}(\kk p, T, \epsilon)$ contains maximal spreading allocations for $T>1/\bar p$, provides a sufficient condition on the asymptotic optimality of $H1$.

\begin{tm}
Let $\epsilon^{*}$ be the optimal value of $H1$. If $T>1/\bar p$, then $\epsilon^{*} = 
{\cal O}(exp({-n}))$.
\end{tm}

\begin{proof}
Let $T>1/\bar p$ and consider the maximal spreading allocation ${\bf x}_{T}^{n}$.
Then, $\epsilon^{*}\leq \epsilon_{s}$, where $\epsilon_{s}$ is the minimum $\epsilon$ such that 
${\bf x}_{T}^{n} \in {\cal H}_{n}(\kk p, T, \epsilon)$.
That is $\epsilon_{s} =  e^{-2n(\bar p - 1/T)^{2}}$,
and since $T>1/\bar p$, $\epsilon^{*}\leq \epsilon_{s} = {\cal O}(exp(-n))$.
\end{proof}

\section{Chernoff Relaxation}
In this section we  take a different approach to obtain a tractable convex relaxation 
for $Q1$ by minimizing an appropriate Chernoff upper bound.

\subsection{Upper Bounding the Objective Function}

\begin{lemma1}(Upper Bound)
Let $Z = \sum_{i=1}^{n}x_{i}Y_{i}$, $x_{i}\geq 0$, $Y_{i} \sim bernoulli(p_{i})$ and $t\geq0$. The probability of failed recovery, 
${\mathbb P}[Z < 1]$, is upper bounded by 
$${\mathbb P}[Z < 1] \leq g_{t}(\kk x)\; =  \sum_{\kk r \subseteq \{1, \dots, n\}}{\mathbb P}(\kk r)
{\mbox{ $\exp$}}\left\{-t \left(  \sum_{i\in \kk r} x_{i}  - 1\right)    \right\}\;\;$$
\label{relax}
\end{lemma1}

\begin{proof}
For any $t\geq0$ we have:
\begin{eqnarray}
\;\;\;{\mathbb P}[Z<1]&\leq& {\mathbb P}[Z \leq 1] \; =\; {\mathbb P}\left[e^{-tZ} \geq e^{-t}\right] \nonumber \\
&\leq& e^{t}{\mathbb E}\left[e^{-tZ} \right] \;=\; e^{t}{\mathbb E}\left[\prod_{i=1}^{n} e^{-tx_{i}Y_{i}} \right]\nonumber \\
&=& e^{t}\prod_{i=1}^{n}{\mathbb E}\left[ e^{-tx_{i}Y_{i}} \right] \nonumber\\
&=& e^{t}\prod_{i=1}^{n} \left(1-p_{i} + p_{i}e^{-tx_{i}} \right) \label{form1}
\end{eqnarray}
\begin{eqnarray}
&=& e^{t}\sum_{\kk r \subseteq \{1, \dots, n\}}{\mathbb P}(\kk r)\;
{\mbox{ $\exp$}}\left\{-t \left(  \sum_{i\in \kk r} x_{i} \right)    \right\}\nonumber \\
&=& \sum_{\kk r \subseteq \{1, \dots, n\}}{\mathbb P}(\kk r)\;
{\mbox{ $\exp$}}\left\{-t \left(  \sum_{i\in \kk r} x_{i} - 1 \right)    \right\}\;\; \;\;\;\;\;\;\label{form2}
\\&\triangleq& g_{t}(\kk x) \nonumber
\end{eqnarray}
\end{proof}

Note that $g_{t}(\kk x)$ is a weighted sum of convex functions with linear arguments, and hence convex in $\kk x$.  Equation (\ref{form2}) makes the convex relaxation of the objective function apparent:

$$\mbox{\Large $\mathbbm{1}$}\Big\{x < \alpha \Big\} \leq e^{-t(x-\alpha)}, \,\mbox{for any $t\geq0$.}$$
 
\subsection{The Relaxed Optimization Problem} 
Before we move forward and state the relaxed optimization problem, we take a closer look at 
the constraint set $S = \{\kk x\in \mathbb R^{n}_{+} : \kk 1^{T}\kk x \leq T \}$ of the original problem $Q1$. From a 
practical perspective, it should be wasteful to allocate more than one unit of data (filesize) on a single node. 
If the node survives, then the data collector can always recover the file using only one unit of data and hence any 
additional storage does  not help. 
Also, an allocation using less than the available budget cannot 
have larger probability of successful recovery.

In the following lemma, we show that it is sufficient to consider allocations with 
$x_{i} \in [0,1]$ and $\sum_{i=1}^{n}x_{i}=T$.

\begin{lemma1}
For any $\kk x \in S $, 
$\exists \kk x' \in S' = \{\kk x\in \mathbb R^{n}_{+} : \kk 1^{T}\kk x = T,\, x_{i}\leq 1,\, i = 1, \dots, n \}$ 
such that ${\mathbb P}\left[\sum_{i=1}^{n}x_{i}'Y_{i}<1\right] \leq {\mathbb P}\left[\sum_{i=1}^{n}x_{i}Y_{i}<1\right]$.
\label{set}
\end{lemma1}

\begin{proof}
See the long version of this paper \cite{ArXiv_version}.
\end{proof}
 
The relaxed  optimization problem can be formulated as follows.

\begin{equation*}
\begin{aligned}
R1: \;\;\;\;\;\;\;\;\;& \underset{\displaystyle x_{i}}{\text{minimize}} 
&  &  \;\;\;\;\;g_{t}(\kk x)  \\
& \text{subject to:}
& & \;\;\;\;\sum_{i=1}^{n}x_{i} = T  \\
& & &\;\;\;\; x_{i} \in  [0, 1] , \; i = 1, \ldots, n.
\end{aligned}
\end{equation*}

Note that, in general, one would like to minimize $\inf_{t\geq0}\{g_{t}(\kk x)\}$ instead of $g_{t}(\kk x)$ for some    $t\geq0$. 
However, for now, we will let $t$ be a free parameter and carry on with the optimization.

The important drawback of the above formulation hides in the objective function: Although convex, $g_{t}(\kk x)$ has an {\it exponentially long}
description in the number of storage nodes: The sum is still over all subsets $\kk r \subseteq \{1, \dots, n\}$. This can be circumvented if we consider minimizing $\log g_{t}(\kk x)$ instead of $g_{t}(\kk x)$ over the same
set.

\begin{lemma1}
$\log g_{t}(\kk x)$ is convex in $\kk x$.
\label{conv}
\end{lemma1}
\begin{proof}  See the long version of this paper \cite{ArXiv_version}.

%
%
%
%
%
\end{proof}

\begin{lemma1}
For any $t\geq0$
$$\displaystyle \arg \min_{\kk x \in S} g_{t}(\kk x) = \arg \min_{\kk x \in S} \sum_{i=1}^{n} \log \left(  1 + \frac{p_{i}}{1-p_{i}}e^{-tx_{i}}  \right), $$
where $S = \{\kk x\in \mathbb R^{n}_{+} : \kk 1^{T}\kk x \leq T, \kk x \preceq \kk 1 \}$. 
\label{equiv}
\end{lemma1}

\begin{proof} Let $\kk x^{*} = \arg \min_{\kk x \in S} g_{t}(\kk x)$. Then $g_{t}(\kk x^{*})\leq g_{t}(\kk x)$, $\forall \kk x \in S$.
Taking the logarithm on both sides preserves the inequality since $\log (\cdot)$ is strictly increasing. 
Hence, $\log g_{t}(\kk x^{*})\leq \log g_{t}(\kk x)$, $\forall \kk x \in S$ and subtracting $t+ \sum_{i=1}^{n}\log(1-p_{i})$ from both sides yields the desired result and completes the proof.
\end{proof}

In view of Lemmas \ref{conv} and \ref{equiv}, we can solve $R1$ through the following {\it equivalent} 
optimization problem.

\begin{equation*}
\begin{aligned}
R2: \;\;\;\;\;\;\;\;& \underset{\displaystyle x_{i}}{\text{minimize}} 
&  &  \;\;\;\;t+\sum_{i=1}^{n} \log \left(  1 + \frac{p_{i}}{1-p_{i}}e^{-tx_{i}}  \right)  \\
& \text{subject to:}
& & \;\;\;\;\sum_{i=1}^{n}x_{i} = T  \\
& & &\;\;\;\; x_{i} \in [0,1] , \; i = 1, \ldots, n.
\end{aligned}
\end{equation*}
$R2$ is a convex {\it separable} optimization problem with  polynomial size description and 
in terms of complexity, it is ``not much harder'' than linear programming \cite{dorit}.  One can solve 
such problems numerically in a very efficient way using standard, ``off-the-shelf'' algorithms and optimization packages
such as $\mathtt{CVX}$ \cite{cvx}, \cite{gb08}.
 
\subsection{Insights from Optimality Conditions for R2}
   
Here, we move one step further and take the KKT conditions for $R2$ in order to take a 
closer look at the structure of the optimal solutions. Let $r_{i} \triangleq \frac{p_{i}}{1-p_{i}}$.

The Lagrangian for $R2$ is:
\begin{eqnarray}
L(\kk x, \kk u, \kk v, \lambda)&=& \sum_{i=1}^{n}\log \left(  1 + r_{i}e^{-tx_{i}}  \right) \nonumber
+ \lambda \left(\sum_{i=1}^{n}x_{i} - T \right) \\ &&\;\;\;\;\;\;\;\;\;\;\;\;\;\;\;\;\;\;- \sum_{i=1}^{n}u_{i}x_{i}  +\sum_{i=1}^{n}v_{i}(x_{i} - 1) \nonumber
\end{eqnarray}
where $\lambda \in {\mathbb R}$, $\kk u,\kk v \in {\mathbb R}_{+}^{n}$ are the corresponding Lagrange multipliers.
The gradient is given by $\displaystyle \nabla_{x_{i}} L(\kk x,\kk u, \kk v, \lambda) = 
-\frac{r_{i}t}{r_{i} + e^{tx_{i}}} +\lambda - u_{i} +v_{i} \, ,$ and the KKT necessary and sufficient conditions for optimality yield:

\begin{eqnarray}
&&-\frac{r_{i}t}{r_{i} + e^{tx^{*}_{i}}} +\lambda - u_{i} +v_{i} = 0 \, , \; \forall i \\
&&\sum_{i=1}^{n}x^{*}_{i} = T                      \label{sum}\\
&&0\leq x^{*}_{i} \leq 1\, , \; \forall i                 \\
&& \lambda \in {\mathbb R}\, ,\; v_{i}, u_{i} \geq 0  \, , \; \forall i \\
&& v_{i}(x_{i}-1)=0 \,,\; u_{i}x_{i}=0  \, , \; \forall i 
\end{eqnarray}

Using the results from \cite{stefanov}, the optimal solution to $R2$ is given by

\begin{equation}
x_{i}^{*}= \left\{ \begin{array}{ll}
0 & \mbox{if $\frac{r_{i}t}{1+r_{i}} \leq \lambda^{*}$} \\
1 & \mbox{if $\lambda^{*} \leq \frac{r_{i}t}{e^{t}+r_{i}}$} \\
\frac{1}{t}\log\left(\frac{r_{i}t}{\lambda^{*}}-r_{i}\right) & \mbox{if $\frac{r_{i}t}{e^{t}+r_{i}} < \lambda^{*} < \frac{r_{i}t}{1+r_{i}}$}
\end{array}
\right. \\
\label{x}
\end{equation}

where $\lambda^{*}$ is chosen such that Eq.(\ref{sum}) is satisfied, i.e,

\begin{eqnarray} 
\;\;\sum_{i=1}^{n}\frac{1}{t}\log\left(\frac{r_{i}t}{\lambda^{*}}-r_{i}\right)
{\mbox{\Large $\mathbbm{1}$}}\left\{ \lambda^{*} \in  \left(\frac{r_{i}t}{e^{t}+r_{i}}, \frac{r_{i}t}{1+r_{i}}\right) \right\} \nonumber
\\+\sum_{i=1}^{n}{\mbox{\Large $\mathbbm{1}$}}\left\{\lambda^{*} \leq  \frac{r_{i}t}{e^{t}+r_{i}}\right\} 
  = T
\label{lambda}
\end{eqnarray} 

Numerically,  $\lambda^{*}$ 
can be computed via an iterative ${\cal O}(n^{2})$ algorithm described in \cite{stefanov}, 
and hence this approach gives an even more efficient way to solve $R2$.

However, the most important aspect of the above result is that we can use equations 
(\ref{x}), (\ref{lambda}) to obtain closed form solutions for a certain region of 
problem parameters and analyze the performance of the resulting allocations.

\subsection{The choice of parameter $t\geq0$}

It is clear that the optimal solution to $R2$ depends on our choice of  $t\geq 0$.
For example,   
$\frac{r_{i}t}{e^{t}+r_{i}}\rightarrow 0$, $\frac{r_{i}t}{1+r_{i}} \rightarrow\infty$, as $t\rightarrow\infty$ and 
$\displaystyle x_{i}^{*} = \lim_{t\rightarrow \infty}{t^{-1}}\log\left({r_{i}t}/{\lambda^{*}}-r_{i}\right)$, $\forall i$. Equation (\ref{lambda}) yields
$x_{i}^{*} = \frac{T}{n} \,, \; \forall i $ 
and hence the maximal spreading allocation becomes  optimal for $R2$ as $t\rightarrow \infty$. 
Even though this motivates the choice of  
maximal spreading allocations as approximate ``one-shot'' solutions for the original problem $Q1$,  
explicitly tuning the parameter $t$ can provide significantly better approximations.

In order to obtain the tightest bound from Lemma \ref{relax}, we have to jointly minimize the objective in $R2$
with respect to $t\geq0$ and $\kk x$. Towards this end, one can {\it iteratively optimize} $R2$ by fixing the value of one variable ($t$ or $\kk x$) at each step and minimizing over the other. After each iteration the objective function decreases and hence the above procedure converges to a (possibly local) minimum. The above algorithm iteratively tunes the Chernoff bound introduced in this section and produces a minimizing allocation that can serve as an approximate solution to the original problem $Q1$. 

For analytic purposes though, we can choose a value for $t$ as follows.
Recall from Lemma \ref{relax} that ${\mathbb P}[Z<1]\leq g_{t}(\kk x)$ for any $t\geq 0$. 
After taking logarithms, we would like to find a value for $t\geq 0$ that minimizes
$b(t)\triangleq t+\sum_{i=1}^{n}\log(1+r_{i}e^{-tx_{i}})$. Notice that $b(t)$ is a convex function of $t$, with $b(t)>0$, $\forall t\geq 0$, 
$b(0)=\sum_{i=1}^{n}\log(1+r_{i})$ and $\lim_{t\rightarrow\infty}b(t)=\infty$. The slope of $b(t)$ at zero is 
 $b'(0) = 1 - \sum_{i=1}^{n}\frac{r_{i}x_{i}}{1+r_{i}}=1-\sum_{i=1}^{n}p_{i}x_{i} $, which is negative if the allocation is reliable.
 
When $t$ is large, $\log(1+r_{i}e^{-tx_{i}})\approx 0$, whereas for small values of $t$, $\log(1+r_{i}e^{-tx_{i}})\approx -tx_{i}+\log r_{i}$ and hence  $b(t)\approx t + \sum_{i=1}^{n}\max\{-tx_{i}+\log r_{i}, 0\}\geq t + \max\{-\sum_{i=1}^{n}tx_{i}+\log r_{i}, 0\}$.  One way to choose $t$ 
that does not depend on $x_{i}$ is to make  $-\sum_{i=1}^{n}tx_{i}+\log r_{i}=0\; \Rightarrow$ 
$  t = \frac{1}{T}\sum_{i=1}^{n}\log r_{i}$.

\subsection{A closed-form allocation: $\hat{ \kk  x}_{T}^{n}$ }
In view of the above results we provide here a closed form allocation (each $x_{i}$ is given as a function of $\kk p$ and T) that can be used to study the asymptotic performance of $R2$
and serve as a better ``one-shot'' approximate solution to Q1.
  
Let ${\cal E}(\cdot)$  be a shorthand notation for the sample average such that ${\cal E}f(x) = \frac{1}{n}\sum_{i=1}^{n}f(x_{i})$, in order to simplify the expressions.
For the above choice of $t = \frac{1}{T}\sum_{i=1}^{n}\log r_{i} = {n{\cal E}{\log r}}/{T} $, equation (\ref{x}) becomes:

\begin{equation}
x_{i}^{*}= \left\{ \begin{array}{ll}
0 & \mbox{if $\frac{r_{i}}{1+r_{i}}\frac{n{\cal E}{\log r}}{T} \leq \lambda^{*}$} \\
1 & \mbox{if $\lambda^{*} \leq \frac{r_{i}{n{\cal E}{\log r}}/{T}}{e^{{n{\cal E}{\log r}}/{T}}+r_{i}}$} \\
\frac{T}{n{\cal E}{\log r}}\log\left(\frac{nr_{i}{\cal E}{\log r}}{T\lambda^{*}}-r_{i}\right) & \mbox{otherwise}
\end{array}
\right. \\
\label{xt}
\end{equation}

\begin{lemma1}
If $p_{i}>\frac{1}{2}$, $\forall i$ and $T < \frac{n{\cal E}{\log r}}{\log r_{max}}$, 
 ${ r_{max}} = max\{r_{i}\}$,
then $x_{i}^{*} = \frac{T}{n{\cal E}{\log r}}\log r_{i}$, $\forall i$.
\end{lemma1}

\begin{proof}
Assume that $\lambda^{*} \in \left(\frac{r_{i}{n{\cal E}{\log r}}/{T}}{e^{{n{\cal E}{\log r}}/{T}}+r_{i}}, \frac{r_{i}}{1+r_{i}}\frac{n{\cal E}{\log r}}{T}\right)$. Then from Eq.(\ref{lambda}), $\lambda^{*} = \frac{n{\cal E}{\log r}}{2T}$ and $x_{i}^{*} = \frac{T}{n{\cal E}{\log r}}\log r_{i}$. $\lambda^{*}$ is indeed in the required interval 
if $\frac{n{\cal E}{\log r}}{2T} < \frac{r_{i}}{1+r_{i}}\frac{n{\cal E}{\log r}}{T},\, \forall i$ $\Rightarrow$ $r_{i}>1,\, \forall i$ $\Rightarrow$ $p_{i}>1/2,\, \forall i$ and 
$\frac{n{\cal E}{\log r}}{2T} < \frac{r_{i}{n{\cal E}{\log r}}/{T}}{e^{{n{\cal E}{\log r}}/{T}}+r_{i}},\, \forall i$ $\Rightarrow$ 
$r_{i}< e^{n{\cal E}{\log r}}/T,\, \forall i$ $\Rightarrow$ $T < \frac{n{\cal E}{\log r}}{\log r_{max}}$.
\end{proof}

Clearly, when all $p_{i}>1/2$, $\hat{ \kk  x}_{T}^{n}$: $x_{i} = \frac{T}{n{\cal E}{\log r}}\log r_{i}$, $\forall i$, is a feasible suboptimal allocation for $Q1$. It is also suboptimal for $R2$ in general, since solving
$R2$ via the proposed algorithms can only achieve a smaller 
probability of failed recovery. We have $P_{e}\{Q1\} \leq P_{e}\{R2\} \leq 
{\mathbb P}\left\{\sum_{i=1}^{n}\frac{T}{n{\cal E}{\log r}}\log r_{i}Y_{i}<1\right\}$.

In the following lemma we give an upper bound on the probability of failed recovery
for $\hat{ \kk  x}_{T}^{n}$ and establish its asymptotic optimality.

\begin{lemma1}
If $p_{i}>\frac{1}{2}$, $\forall i$ and $\displaystyle T > \frac{{\cal E}{\log r}}{{\cal E}{p \log r}}$, 
the allocation $\hat{ \kk  x}_{T}^{n} : x_{i} = \frac{T}{n{\cal E}{\log r}}\log r_{i},$ $\forall i$, \,is strictly 
reliable, and 
the probability of failed recovery, $P_{e}={\mathbb P}[Z<1]$, is upper bounded by
$$P_{e}\leq 
exp \left\{  -2n \frac{\left({\cal E}{p\log r} - \frac{{\cal E}{\log{r}}}{T}\right)^{2}}{{\cal E}{\log^{2}r}}\right\}$$
and hence, when $\displaystyle T > \frac{{\cal E}{\log r}}{{\cal E}{p \log r}}$, $P_{e}\rightarrow 0$, as $n\rightarrow \infty$.
\label{lem:chern}
\end{lemma1}

\begin{proof}

The proof follows directly from Lemma \ref{lem1} and Equation (\ref{bound}).
\end{proof}

Notice that $\hat{ \kk  x}_{T}^{n}$ is reliable for values of $T$ for which a maximal spreading allocation 
${ \kk  x}_{T}^{n}$ is not, since 
$\frac{1}{\bar p}\geq \frac{{\cal E}{\log r}}{{\cal E}{p \log r}}$, and hence its probability of failed recovery $P_{e}$ goes to zero exponentially fast for smaller values of $T$.

\section{Numerical Experiments}

In this section we evaluate the performance of the proposed 
approximate distributed storage allocations in terms of their probability of failed recovery and plot the corresponding bounds.
In our simulations we consider an ensemble of distributed 
storage systems with $n=100$ nodes, in which the corresponding access probabilities, $p_{i}\sim {\cal U}(0.5, 1)$, are drawn uniformly at random from the interval $(0.5, 1)$.

We consider the following allocations. 
1) {\it Maximal spreading}: $x_{i} = \frac{T}{n}$,  $\forall i$.
2) {\it Chernoff closed-form}: $x_{i} = ({T}/{n{\cal E}{\log r}})\log r_{i},$ $\forall i$.
3) {\it Hoeffding $\epsilon$-optimal}: obtained by solving $H1$.
4) {\it Chernoff iterative}: obtained by solving $R2$ and iteratively tuning the parameter $t$.

Fig.\ref{fig1} shows, in solid lines, the ensemble average probability of failed recovery of each allocation, 
${\mathbb P}\left[\sum_{i=1}^{n}x_{i}Y_{i}<1\right]$, versus the maximum available budget $T$.
In dashed lines, Fig.\ref{fig1} plots the corresponding bounds on $P_{e}$ obtained from Corollary \ref{cor:maxspread}, Lemma \ref{lem:chern} and the objective functions of $H1$, $R1$.

\begin{figure}
\centerline{\includegraphics[width=1.05\columnwidth]{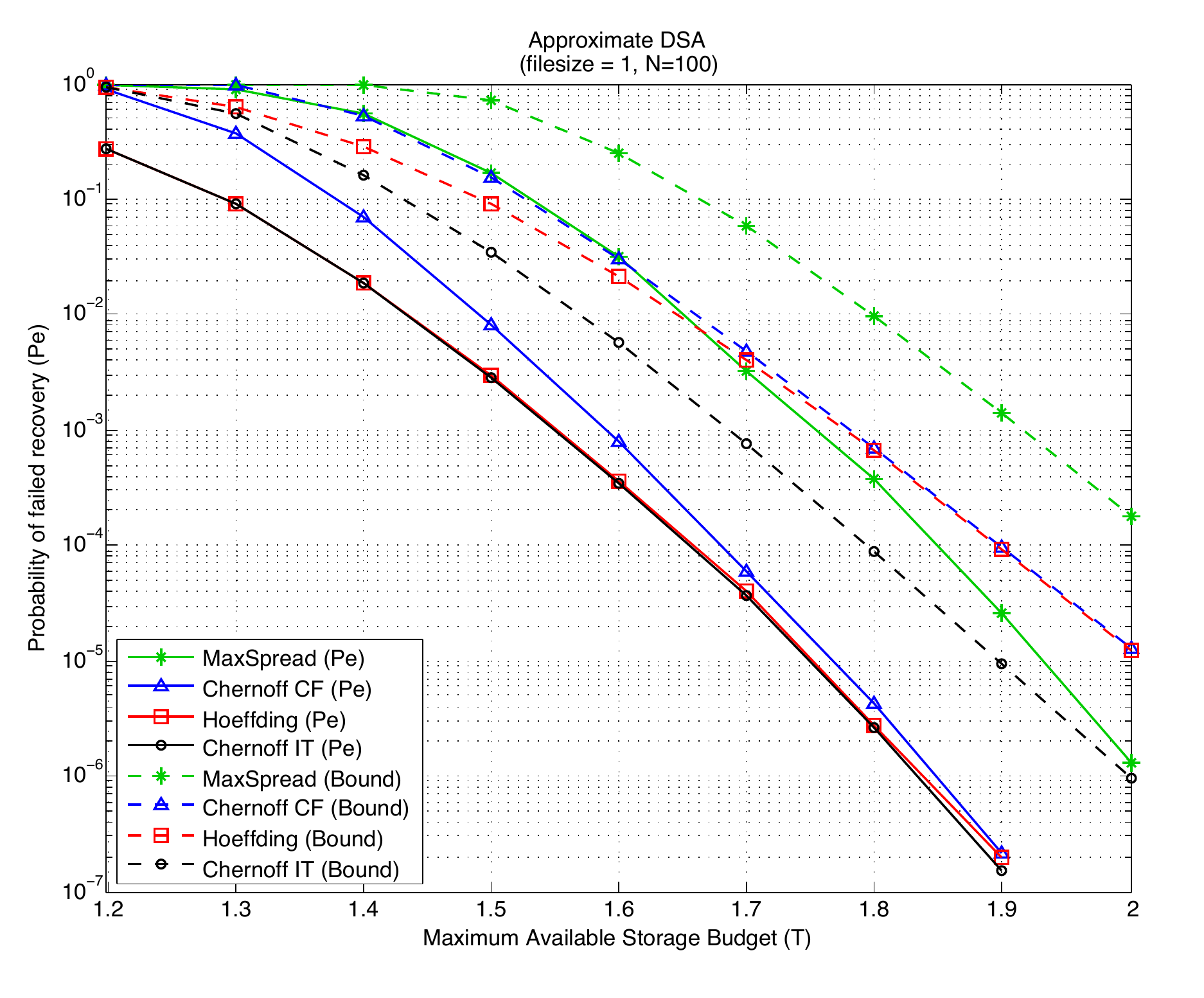}}
\caption{Performance of the proposed approximate distributed storage allocations and their corresponding upper bounds for a system with $n=100$ nodes and $p_{i}\sim {\cal U}(0.5, 1)$.}
\label{fig1}
\end{figure}

%
%
%
%

\bibliographystyle{unsrt}
\bibliography{references}

\begin{thebibliography}{1}

\bibitem{alloc}
D.~Leong, A.~Dimakis, and T.~Ho.
\newblock Distributed storage allocations.
\newblock {\em CoRR}, abs/1011.5287, 2010.

\bibitem{Hoeffding}
W.~Hoeffding.
\newblock Probability inequalities for sums of bounded random variables.
\newblock {\em Journal of the American Stat. Association}, 58(301):13--30,
  March 1963.

\bibitem{mitzbook}
M.~Mitzenmacher and E.~Upfal.
\newblock {\em Probability and Computing: Randomized Algorithms and
  Probabilistic Analysis}.
\newblock Cambridge University Press, New York, NY, USA, 2005.

\bibitem{ArXiv_version}
V.~Ntranos, G.~Caire, and A.~Dimakis.
\newblock Allocations for heterogenous distributed storage ({\it long
  version}).
\newblock \url{http://www-scf.usc.edu/~ntranos/docs/HDS-long.pdf}, January
  2012.

\bibitem{dorit}
D.~S. Hochbaum and J.~George Shanthikumar.
\newblock Convex separable optimization is not much harder than linear
  optimization.
\newblock {\em J. ACM}, 37:843--862, October 1990.

\bibitem{cvx}
M.~Grant and S.~Boyd.
\newblock {CVX}: Matlab software for disciplined convex programming, version
  1.21.
\newblock {http://cvxr.com/cvx}, April 2011.

\bibitem{gb08}
M.~Grant and S.~Boyd.
\newblock Graph implementations for nonsmooth convex programs.
\newblock In V.~Blondel, S.~Boyd, and H.~Kimura, editors, {\em Recent Advances
  in Learning and Control}, Lecture Notes in Control and Information Sciences,
  pages 95--110. Springer-Verlag Limited, 2008.

\bibitem{stefanov}
S.~M. Stefanov.
\newblock Convex separable minimization subject to bounded variables.
\newblock {\em Comp. Optimization and Applications}, 18, 2001.

\end{thebibliography}

\end{document}